\documentclass[runningheads]{llncs}
\usepackage{amsmath,amssymb,ifthen}
\usepackage{graphicx,rotating,paralist}
\usepackage[font=small]{subfig,caption}
\usepackage{placeins,rotating,hyperref}
\graphicspath{{img/}}

\newcommand{\rephrase}[3]{\noindent\textbf{#1~#2.}\hspace{1.5ex}\emph{#3}\medskip}
\newcommand{\algostep}[1]{\smallskip\noindent{\textit{#1}.}}

\newcommand{\ver}{arxiv}
\newcommand{\arxapp}[2]{\ifthenelse{\equal{\ver}{arxiv}}{#2}{#1}}

\author{Michael A. Bekos, Henry F\"orster, Michael Kaufmann}
\title{On Smooth Orthogonal and Octilinear Drawings: Relations, Complexity and Kandinsky Drawings\thanks{This work is supported by DFG grant Ka812/17-1.}}
\titlerunning{On Smooth Orthogonal and Octilinear Drawings}

\institute{
Wilhelm-Schickhard-Institut f\"ur Informatik, Universit\"at T\"ubingen, Germany
}

\begin{document}
\maketitle

\begin{abstract}
We study two variants of the well-known orthogonal drawing model: (i)~the smooth orthogonal, and (ii)~the octilinear. Both models form an extension of the orthogonal, by supporting one additional type of edge segments (circular arcs and diagonal segments, respectively).

For planar graphs of max-degree~$4$, we analyze relationships between the graph classes that can be drawn bendless in the two models and we also prove NP-hardness for a restricted version of the bendless drawing problem for both models. For planar graphs of higher degree, we present an algorithm that produces bi-monotone smooth orthogonal drawings with at most two segments per edge, which also guarantees a linear number of edges with exactly one segment.
\end{abstract}

\section{Introduction}
\label{sec:introduction}

Orthogonal graph drawing is an intensively studied and well established model~for drawing graphs. As a result, several efficient algorithms providing good aesthetics and good readability have been proposed over the years, see e.g.,~\cite{kant,kandinsky,liuEtAl,tamassia}. In such drawings, each vertex corresponds to a point on the Euclidean plane and each edge is drawn as a sequence of axis-aligned line segments; see Fig.~\ref{fig:introduction}.

Several research directions build upon this successful model. We focus on two models that have recently received attention: %
\begin{inparaenum}[(i)]
\item the \emph{smooth orthogonal}~\cite{smog1}, in which every edge is a sequence of axis-aligned segments and circular arc segments with common axis-aligned tangents (i.e., quarter, half or three-quarter arc segments), and
\item the \emph{octilinear}~\cite{octi}, in which every edge is a sequence of axis-aligned and diagonal (at $\pm 45^\circ$) segments.
\end{inparaenum}

Observe that both models extend the orthogonal by allowing one more type of edge-segments. The former was introduced with the aim of combining the artistic appeal of \emph{Lombardi drawings}~\cite{lombardi-1,lombardi-2} with the clarity of the orthogonal drawings. The latter, on the other hand, is primarily motivated by metro-map and map schematization applications (see, e.g.,~\cite{metromaps-1,octiNP,metromaps-2,metromaps-3}). Note that in the orthogonal and in the smooth orthogonal models, each edge may enter a vertex using one out of four available (axis-aligned) directions, called \emph{ports}. Thus both models support graphs of max-degree~$4$. In the octilinear model, each vertex has eight available ports and therefore one can draw graphs of max-degree~$8$.

For readability purposes, usually in such drawings one seeks to minimize the \emph{edge complexity}~\cite{gd-book-1,gd-book-2}, i.e., the maximum number of segments used for representing any edge. Also, when the input is a planar graph, one seeks for a corresponding planar drawing. Note that drawings with edge complexity~$1$ are also called \emph{bendless}. We refer to drawings with edge complexity $k$ as~\emph{$k$-drawings}; thus, by definition, orthogonal $k$-drawings have at most $k-1$ bends per edge.

\begin{figure}[t!]
	\centering
	\begin{minipage}[b]{.2\textwidth}
		\centering
		\subfloat[\label{fig:straight} {}]{
		\includegraphics[width=.8\textwidth,page=1]{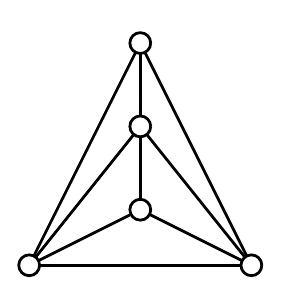}}
	\end{minipage}
	\hfil
	\begin{minipage}[b]{.2\textwidth}
		\centering
		\subfloat[\label{fig:orthogonal} {}]{
		\includegraphics[width=.8\textwidth,page=2]{introduction}}
	\end{minipage}
	\hfil
	\begin{minipage}[b]{.2\textwidth}
	\centering
		\subfloat[\label{fig:octilinear} {}]{
		\includegraphics[width=.8\textwidth,page=3]{introduction}}
	\end{minipage}
	\hfil
	\begin{minipage}[b]{.2\textwidth}
	\centering
		\subfloat[\label{fig:smooth} {}]{
		\includegraphics[width=.8\textwidth,page=4]{introduction}}
	\end{minipage}
	\caption{Different drawings of a planar graph of max-degree~$4$:
	(a)~straight-line,
	(b)~orthogonal $3$-drawing,
	(c)~octilinear $2$-drawing, and
	(d)~smooth orthogonal $2$-drawing.}
	\label{fig:introduction}
\end{figure}

\paragraph{Known results.} There exists a plethora of results for each of the aforementioned models; here we list existing results for drawings with low edge complexity.
\begin{itemize}[-]
\item All planar graphs of max-degree~$4$, except for the octahedron, admit orthogonal $3$-drawings; the octahedron is orthogonal $4$-drawable~\cite{kant,liuEtAl}. Minimizing the number of bends over all embeddings of a planar graph of max-degree~$4$ is $\mathcal{NP}$-hard~\cite{gargTamassia}. For a given planar embedding, however, finding a planar orthogonal drawing with minimum number of bends can be done in polynomial time by an approach, called \emph{topology-shape-metrics}~\cite{tamassia}, that is based on min-cost flow computations and works in three phases. Initially, a planar embedding is computed if not specified by the input. In the next phase, the angles and the bends of the drawing are computed, yielding an \emph{orthogonal representation}. In the last phase, the actual coordinates for the vertices and bends are computed.
\item All planar graphs of max-degree~$4$ (including the octahedron) admit smooth orthogonal $2$-drawings. Note that not all planar graphs of max-degree~$4$ allow for bendless smooth orthogonal drawings~\cite{smog1}, and that such drawings may require exponential area~\cite{smog2}. Bendless smooth orthogonal drawings are possible only for subclasses, e.g., for planar graphs of max-degree~$3$~\cite{perfectSmog} and for outerplanar graphs of max-degree~$4$~\cite{smog2}. It is worth mentioning that the complexity of the problem, whether a planar graph of max-degree~$4$ admits a bendless smooth orthogonal drawing, has not been settled (it is conjectured to be $\mathcal{NP}$-hard~\cite{smog2}).
\item All planar graphs of max-degree~$8$ admit octilinear $3$-drawings~\cite{slopes}, while planar graphs of max-degree~$4$ or $5$ allow for octilinear $2$-drawings~\cite{octi}. Bendless octilinear drawings are always possible for planar graphs of max-degree~$3$~\cite{latin}. Note that deciding whether an embedded planar graph of max-degree~$8$ admits a bendless octilinear drawing is $\mathcal{NP}$-hard~\cite{octiNP}. It is not, however, known whether this negative result applies for planar graphs of max-degree~$4$ or whether these graphs allow for a decision algorithm (in fact, there exist planar graphs of max-degree~$4$ that do not admit bendless octilinear drawings~\cite{octi-2}).
\end{itemize}

\paragraph{Our contribution} Motivated by the fact that usually one can ``easily'' convert an octilinear drawing of a planar graph of max-degree~$4$ to a corresponding smooth orthogonal one (e.g., by replacing diagonal edge segments with quarter circular arc segments; see Figs.~\ref{fig:octilinear}-\ref{fig:smooth} for an example), and vice versa, we study in Section~\ref{sec:relationships} inclusion-relationships between the graph-classes that admit such drawings. In Section~\ref{sec:nphardness}, we show that it is $\mathcal{NP}$-hard to decide whether an embedded planar graph of max-degree~$4$ admits a bendless smooth orthogonal or a bendless octilinear drawing, in the case where the angles between any two edges incident to a common vertex and the shapes of all edges are specified as part of the input (e.g., as in the last step of the topology-shape-metrics approach~\cite{tamassia}). Our proof is a step towards settling the complexities of both decision problems in their general form. Inspired from the \emph{Kandinsky model}~(see, e.g.,~\cite{kandinsky-2,kandinsky-3,kandinsky}) for drawing planar graphs of arbitrary degree in an orthogonal style, we present in Section~\ref{sec:kandinsky} two drawing algorithms that yield bi-monotone smooth orthogonal drawings of good quality. The first yields drawings of smaller area, which can also be transformed to octilinear with bends at $135^\circ$. The second yields larger drawings but guarantees that at most $2n-5$ edges are drawn with two segments. We conclude in Section~\ref{sec:conclusions} with open problems.

\begin{figure}[t!]
	\centering
	\includegraphics[scale=0.65,page=1]{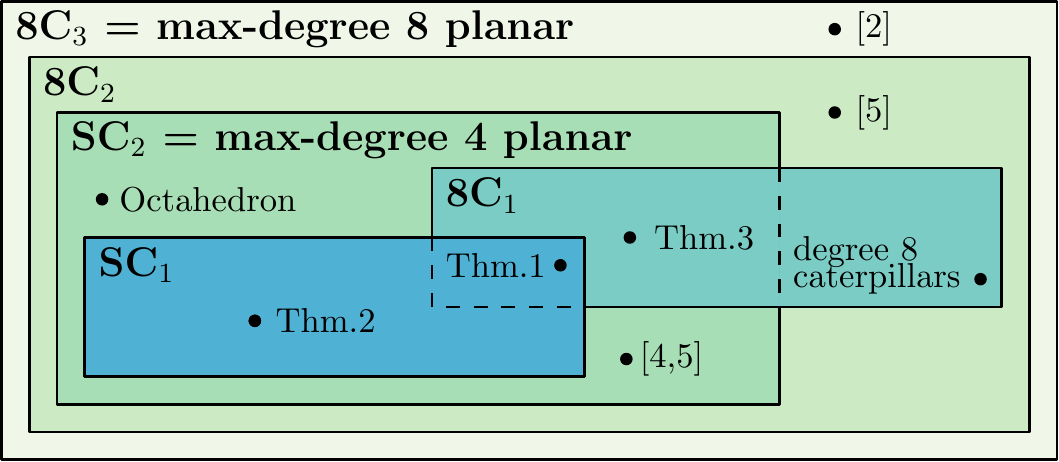}
	\caption{%
	Different inclusion-relationships: For $k \geq 1$, $SC_k$ and $8C_k$ correspond to the classes of graphs admitting smooth orthogonal and octilinear $k$-drawings, respectively.
	}
	\label{fig:graphClasses}
\end{figure}

\paragraph{Preliminaries} For graph theoretic notions refer to~\cite{graphtheory}. For definitions on planar graphs, we point the reader to~\cite{gd-book-1,gd-book-2}. We also assume familiarity with standard graph drawing techniques, such as the \emph{canonical ordering}~\cite{dFPP,cannonicalorder} and the \emph{shift-method} by de~Fraysseix, Pach and Pollack~\cite{dFPP}; see~\arxapp{\cite{arxiv}}{App.~\ref{app:preliminaries}} for more details.

\section{Relationships between Graph Classes}
\label{sec:relationships}

In this section, we consider relationships between the classes of graphs that admit smooth orthogonal $k$-drawings and octilinear $k$-drawings, $k \geq 1$, denoted as $SC_k$ and $8C_k$, respectively. Our findings are also summarized in Fig.~\ref{fig:graphClasses}.

By definition, $SC_1 \subseteq SC_2$ and $8C_1 \subseteq 8C_2 \subseteq 8C_3$ hold. Since each planar graph of max-degree~$8$ admits an octilinear $3$-drawing~\cite{slopes}, class $8C_3$ coincides with the class of planar graphs of max-degree~$8$. Similarly, class $SC_2$ coincides with the class of planar graph of max-degree~$4$, as these graphs admit smooth~orthogonal $2$-drawings~\cite{smog2}. This also implies that $SC_2 \subseteq 8C_2$, since each planar graph of max-degree~$4$ admits an octilinear $2$-drawing~\cite{octi}. The relationship $8C_2 \neq 8C_3$ follows from~\cite{octi}, where it was proven that there exist planar graphs of max-degree~$6$ that do not admit octilinear $2$-drawings. The relationship $SC_2 \neq 8C_2$ follows from~\cite{octi-2}, where it was shown that there exist planar graphs of max-degree~$5$ that admit octilinear $2$-drawings and no octilinear $1$-drawings, and the fact that planar graphs of max-degree~$5$ cannot be drawn in the smooth orthogonal model. The octahedron graph admits neither a bendless smooth orthogonal drawing~\cite{smog1} nor a bendless octilinear drawing~\cite{octi-2}. However, since it is of max-degree~$4$, it admits $2$-drawings in both models~\cite{smog2,octi}. Hence, it belongs to $8C_2 \cap SC_2 \setminus (8C_1 \cup SC_1)$. To prove that $8C_1 \setminus SC_2 \neq \emptyset$, observe that a caterpillar whose spine vertices are of degree 8 clearly admits an octilinear $1$-drawing, however, due to its degree it does not admit a smooth orthogonal.

To complete the discussion of the relationships of Fig.~\ref{fig:graphClasses}, we have to show~that $SC_1$ and $8C_1$ are incomparable. This is the most interesting part of our proof, as usually one can ``easily'' convert a bendless octilinear drawing of a planar graph of max-degree~$4$ to a corresponding bendless smooth orthogonal one (e.g., by replacing diagonal segments with quarter circular arcs), and vice versa; see, e.g., Figs.~\ref{fig:octilinear}-\ref{fig:smooth}. Since the endpoints of each edge of a bendless smooth orthogonal or octilinear drawing are along a line with slope $0$, $1$, $-1$ or~$\infty$, such conversions are in principle possible. Two difficulties that might arise are to preserve planarity and to guarantee that no two edges enter a vertex using the same port. Clearly, however, there exist infinitely many (even $4$-regular) planar graphs that admit both drawings in both models; see Fig.~\ref{fig:trains} and \arxapp{\cite{arxiv}}{App.~\ref{app:relationships}} for more details.

\begin{figure}[t!]
	\centering
	\begin{minipage}[b]{0.45\textwidth}
		\centering
		\subfloat[\label{fig:trainsSC1} {A smooth orthogonal $1$-drawing}]{
		\includegraphics[scale=0.63,page=3]{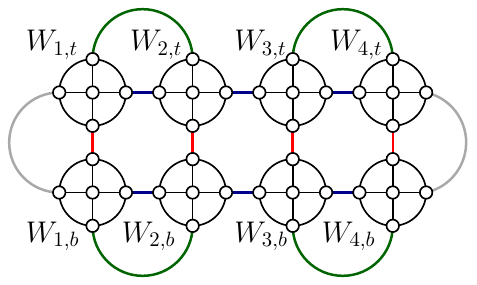}}
	\end{minipage}
	\hfil
	\begin{minipage}[b]{0.45\textwidth}
		\subfloat[\label{fig:trains8C1} {An octilinear $1$-drawing}]{
		\includegraphics[scale=0.63,page=2]{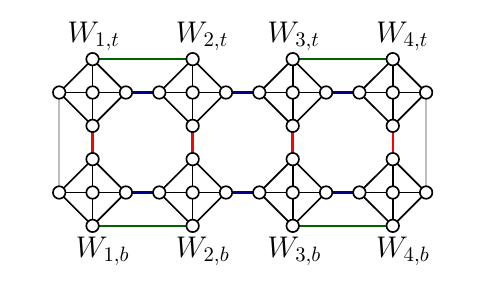}}
	\end{minipage}
	\caption{Illustrations for the proof of Theorem~\ref{theo:union}.}
\label{fig:trains}
\end{figure}

\newcommand{\union}{There is an infinitely large family of $4$-regular planar graphs that admit both bendless smooth orthogonal and bendless octilinear drawings.}
\begin{theorem}
\union
\label{theo:union}
\end{theorem}

\noindent In the next two theorems we show that $SC_1$ and $8C_1$ are incomparable.

\begin{theorem}
There is an infinitely large family of $4$-regular planar graphs that admit bendless smooth orthogonal drawings but no bendless octilinear drawings.
\label{theo:smoothAndNotOcti}
\end{theorem}
\begin{proof}
Consider the planar graph $C$ of Fig.~\ref{fig:4gonOctahedronComponentSC1}, which is drawn bendless smooth orthogonal. We claim that $C$ admits no bendless octilinear drawing. If one substitutes its degree-$2$ vertex (denoted by~$c$ in Fig.~\ref{fig:4gonOctahedronComponentSC1}) by an edge connecting its two neighbors, then the resulting graph is triconnected, which admits an unique embedding (up to the choice of its outerface; see Figs.~\ref{fig:4gonOctahedronComponentSC1}-\ref{fig:4gonOctahedronComponentEmbedding2}). Now, observe that the outerface of any octilinear drawing of graph $C$ (if any) has length at most $5$ (Constraint~$1$). In addition, each vertex of this outerface (except for $c$, which is of degree~$2$) must have two ports pointing in the interior of this drawing, because every vertex of $C$ is of degree~$4$ except for $c$. This implies that the angle formed by any two consecutive edges of this outerface is at most $225^\circ$, except for the pair of edges incident to $c$ (Constraint~$2$). But if we want to satisfy both constraints, then at least one edge of this outerface must be drawn with a bend; see Fig.~\ref{fig:4gonOctahedronComponentNotOcC1}. Hence, graph $C$ does not admit a bendless octilinear drawing.

\begin{figure}[t!]
	\centering
	\begin{minipage}[b]{0.18\textwidth}
		\centering
		\subfloat[\label{fig:4gonOctahedronComponentSC1} {}]{
		\includegraphics[scale=0.5,page=5]{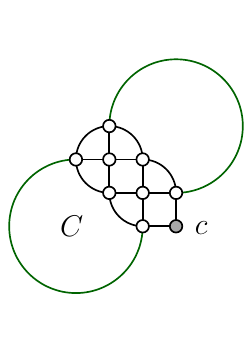}}
	\end{minipage}
	\hfil
	\begin{minipage}[b]{0.18\textwidth}
		\centering
		\subfloat[\label{fig:4gonOctahedronComponentEmbedding2} {}]{
		\includegraphics[scale=0.5,page=6]{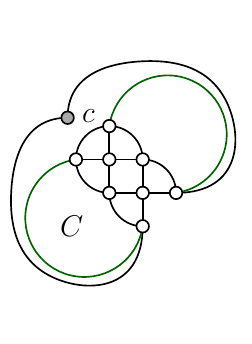}}
	\end{minipage}
	\hfil
	\begin{minipage}[b]{0.18\textwidth}
		\centering
		\subfloat[\label{fig:4gonOctahedronComponentNotOcC1} {}]{
		\includegraphics[scale=0.5,page=7]{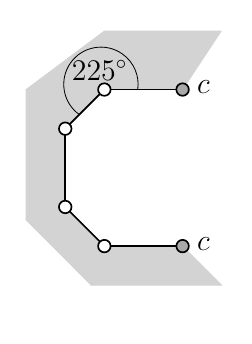}}
	\end{minipage}
	\hfil
	\begin{minipage}[b]{0.36\textwidth}
		\centering
		\subfloat[\label{fig:4gonOctahedronChain} {}]{
		\includegraphics[scale=0.5,page=4]{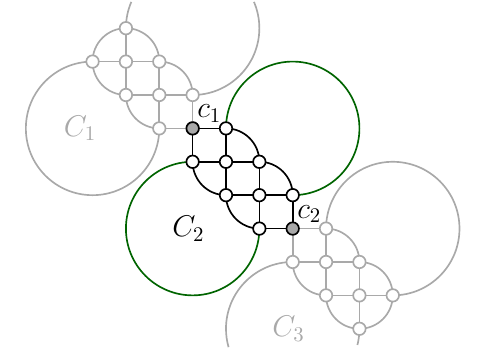}}
	\end{minipage}
	\caption{%
	Illustrations for the proof of Theorem~\ref{theo:smoothAndNotOcti}.}
\label{fig:4gonOctahedron}
\end{figure}

Based on graph $C$, for each $k \in \mathbb{N}_0$ we construct a $4$-regular planar graph $G_k$ consisting of $k + 2$ biconnected components $C_1,\ldots,C_{k+2}$ arranged in a chain; see Fig.~\ref{fig:4gonOctahedronChain} for the case $k=1$. Clearly, $G_k$ admits a bendless smooth orthogonal drawing for any $k$. Since the end-components of the chain (i.e., $C_1$ and $C_{k+2}$)~are isomorphic to $C$, $G_k$ does not admit a bendless octilinear drawing for any $k$.\qed
\end{proof}

\begin{theorem}
There is an infinitely large family of $4$-regular planar graphs that admit bendless octilinear drawings but no bendless smooth orthogonal drawings.
\label{theo:octiAndNotSmooth}
\end{theorem}
\begin{proof}[sketch]
Consider the planar graph $B$ of Fig.~\ref{fig:necklaceGem}, which is drawn bendless octilinear. Graph $B$ has two separation pairs (i.e., $\{t_1,t_2\}$ and $\{p_1,p_2\}$ in Fig.~\ref{fig:necklaceGem}).

\begin{figure}[h]
	\centering
	\subfloat[\label{fig:necklaceGem} {}]{
	\includegraphics[scale=0.38,page=8]{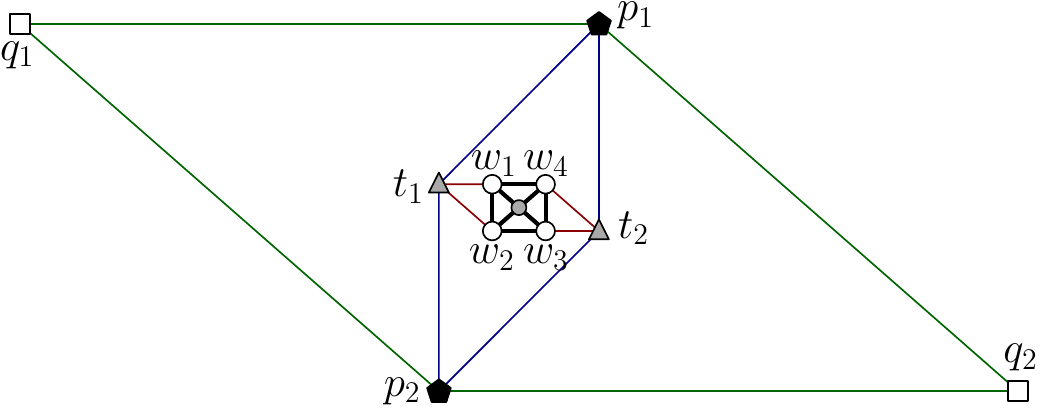}}
	\hfil
	\subfloat[\label{fig:necklaceConstruction} {}]{
	\includegraphics[scale=0.38,page=9]{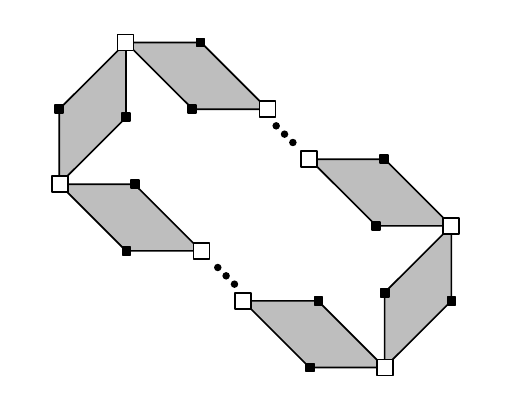}}
	\caption{Illustrations for the proof of Theorem~\ref{theo:octiAndNotSmooth}.}
	\label{fig:necklace}
\end{figure}

Based on graph $B$, for each $k \in \mathbb{N}_0$ we construct a $4$-regular planar graph $G_k$ consisting of $2k + 4$ copies of $B$ arranged in a cycle; see Fig.~\ref{fig:necklaceConstruction} where each copy of $B$ is drawn as a gray-shaded parallelogram. By construction, $G_k$ admits a bendless octilinear drawing for any $k$. By planarity at least one copy of graph $B$ must be embedded with the outerface of Fig.~\ref{fig:necklaceGem}. However, if we require the outerface of $B$ to be the one of Fig.~\ref{fig:necklaceGem}, then all possible planar embeddings of $B$ are isomorphic to the one of Fig.~\ref{fig:necklaceGem}. We exploit this property in \arxapp{\cite{arxiv}}{App.~\ref{app:relationships}} to show that $B$ does not admit a bendless smooth~orthogonal drawing with this outerface. The detailed proof is based on an exhaustive consideration of all bendless smooth orthogonal drawings of subgraphs of $B$, which we incrementally augment by adding more vertices~to~them. Thus, for any~$k$, graph $G_k$ does not admit a bendless smooth orthogonal drawing.\qed
\end{proof}

\section{$\mathcal{NP}$-hardness Results}
\label{sec:nphardness}

In this section, we study the complexity of the bendless smooth orthogonal and octilinear drawing problems. As a first step towards addressing the complexity of both problems for planar graphs of max-degree~$4$ in general, here we make an additional assumption. We assume that the input, apart from an embedding, also specifies a \emph{smooth orthogonal} or an \emph{octilinear representation}, which are defined analogously to the orthogonal ones: %
\begin{inparaenum}[(i)]
\item the angles between consecutive edges incident to a common vertex in the cyclic order around it (given by the planar embedding) are specified, and
\item the \emph{shape} of each edge (e.g.,~straight-line, or quarter-circular arc) is also specified.
\end{inparaenum}
In other words, we assume that our input is analogous to the one of the last step of the topology-shape-metrics approach~\cite{tamassia}. 

\begin{theorem}\label{thm:np-smooth}
Given a planar graph $G$ of max-degree~$4$ and a smooth orthogonal representation $\mathcal{R}$, it is $\mathcal{NP}$-hard to decide whether $G$ admits a bendless smooth orthogonal drawing preserving $\mathcal{R}$.
\end{theorem}
\begin{proof}Our reduction is from the well-known $3$-SAT problem~\cite{GareyJ79}. Given a formula $\varphi$ in conjunctive normal form, we construct a graph $G_\varphi$ and a smooth orthogonal representation $\mathcal{R}_\varphi$, such that $G_\varphi$ admits a bendless smooth orthogonal drawing $\Gamma_\varphi$ preserving $\mathcal{R}_\varphi$ if and only if $\varphi$ is satisfiable; see also Fig.~\ref{fig:example}.

The main ideas of our construction are: %
\begin{inparaenum}[(i)]
\item specific straight-line edges in~$\Gamma_\varphi$ transport \emph{information} encoded in their length,
\item rectangular faces of $\Gamma_\varphi$ propagate the edge length of one side to its opposite, and
\item for a face composed~of two straight-line edges and a quarter circle arc, the straight-line edges are of~same length, which allows us to change the \emph{direction} in which the information~``flows''.
\end{inparaenum}

\begin{sidewaysfigure}
	\centering
	\includegraphics[scale=0.2,page=14]{smogNPFigures}
	\caption{Drawing $\Gamma_\varphi$ for $\varphi = (a \vee b \vee c) \wedge (\overline{a} \vee
	\overline{b} \vee c)$ and the assignment $a = \texttt{false}$ and $b = c = \texttt{true}$.}
	\label{fig:example}
\end{sidewaysfigure}

\algostep{Variable gadget} For each variable~$x$ of $\varphi$, we introduce a gadget; see Figs.~\ref{fig:variable-1}-\ref{fig:variable-2}. The bold-drawn quarter circle arc  ensures that the sum of the edge lengths to its left is the same as the sum of the edge lengths to its bottom (refer to the edges with gray endvertices). As ``input'' the gadget gets three edges of unit length~$\ell(u)$. This ensures that $\ell(x) + \ell(\overline{x}) = 3 \cdot \ell(u)$ holds for the ``output literals'' $x$ and $\overline{x}$, where $\ell(x)$ and $\ell(\overline{x})$ denote the lengths of two edges representing $x$ and~$\overline{x}$.

\begin{figure}[t]
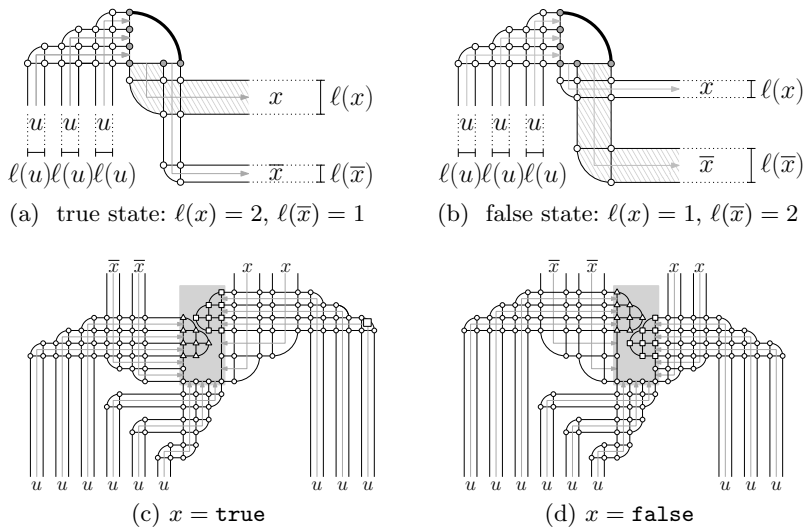

\centering
	\subfloat[\label{fig:variable-1} {true state: $\ell(x)=2$, $\ell(\overline{x})=1$}]{
	\includegraphics[scale=0.4,page=2]{smogNPFigures}}
	\hfil
	\subfloat[\label{fig:variable-2} {false state: $\ell(x)=1$, $\ell(\overline{x})=2$}]{
	\includegraphics[scale=0.4,page=1]{smogNPFigures}}

	\subfloat[\label{fig:parity-1}{$x=\mathtt{true}$}]{
	\includegraphics[scale=0.3,page=6]{smogNPFigures}}
	\hfil
	\subfloat[\label{fig:parity-2}{$x=\mathtt{false}$}]{
	\includegraphics[scale=0.3,page=7]{smogNPFigures}}
	\caption{The (a)-(b)~variable, and (c)-(d)~the parity gadgets; gray-colored arrows show the information ``flow''.}
	\label{fig:gadgets}
\end{figure}

To introduce our concept, assume that the lengths of all straight-line edges are integral and at least~$1$. If we could require $\ell(u) = 1$, then $\ell(x), \ell(\overline{x}) \in \{1,2\}$. This would allow us to encode the assignment $x=\mathtt{true}$~with $\ell(x) = 2$ and $\ell(\overline{x}) = 1$, and the assignment $x=\mathtt{false}$ with $\ell(x) = 1$ and $\ell(\overline{x}) = 2$.
However, if we cannot avoid, e.g., that $\ell(u) = 2$, then the variable gadget would not prevent us from setting $\ell(x) = \ell(\overline{x}) = 3$, which means that $x$ and $\overline{x}$ are ``half-true''. We solve this issue by the so-called \emph{parity gadget}, that allows us to relax the integral constraint and to ensure that $\ell(x), \ell(\overline{x}) \in \{\ell(u) + \varepsilon, 2 \ell(u) - \varepsilon \}$, for $0 < \varepsilon << \ell(u)$.

\algostep{Parity gadget} For each variable~$x$ of $\varphi$, $G_\varphi$ has a gadget (see Figs.~\ref{fig:parity-1}-\ref{fig:parity-2}), which results in overlaps in $\Gamma_\varphi$, if the values of $\ell(x)$ and $\ell(\overline{x})$ do not differ significantly. The central part of this gadget is a ``\emph{vertical gap}'' of width $3 \cdot \ell(u)$ (shaded in gray in Figs.~\ref{fig:parity-1}-\ref{fig:parity-2}) with two blocks of vertices (triangular- and square-shaped in Figs.~\ref{fig:parity-1}-\ref{fig:parity-2}) pointing inside the gap. Each block defines two square-shaped faces and three faces of length~$3$, each formed by two straight-line edges and a quarter circle arc. Depending on the choice of $\ell(x)$ and $\ell(\overline{x})$, one of the blocks may be located above the other. If $\ell(x) \approx \ell(\overline{x})$, however, we can observe that the two blocks are not far enough apart from each other, which leads to overlaps. Using elementary geometry, we prove in \arxapp{\cite{arxiv}}{App.~\ref{app:np-hardness}} that overlaps can be avoided if and only if $|\ell(\overline{x}) - \ell(x)|> \sqrt{3}/2 \cdot \ell(u)\approx 0.866 \cdot \ell(u)$, which implies:
that $\ell(x), \ell(\overline{x}) \in (0,~1.067 \cdot \ell(u)] \cup [1.933 \cdot \ell(u),~3)\text{, i.e., }\varepsilon < 0.067 \cdot \ell(u)$.

\algostep{Clause gadget} For each clause of $\varphi$ with literals $a$, $b$ and $c$, we introduce a~gadget, which is illustrated in Fig.~\ref{fig:clauseGadget}. The bold-drawn quarter circle arc of Fig.~\ref{fig:clauseGadget} compares two sums of information. From the righthand side, four edges of unit length ``enter'' the arc. Observe that there is also a \emph{free edge} (marked with an asterisk in Fig.~\ref{fig:clauseGadget}), which also contributes to the sum but can be stretched independently of any other edge. Hence, the sum of edge lengths on the righthand side of this arc is $> 4 \cdot \ell(u)$. The three literals ``enter'' at the bottom; the sum here is $\ell(a) + \ell(b) + \ell(c)$. Combining both, we obtain that $\ell(a) + \ell(b) + \ell(c) > 4 \cdot \ell(u)$ must hold. This implies that not all $a$, $b$ and $c$ can be $\texttt{false}$, since in this case $\ell(a) + \ell(b) + \ell(c) = 3 \cdot (\ell(u)+\varepsilon) < 4 \cdot \ell(u)$. 

\begin{figure}[t]
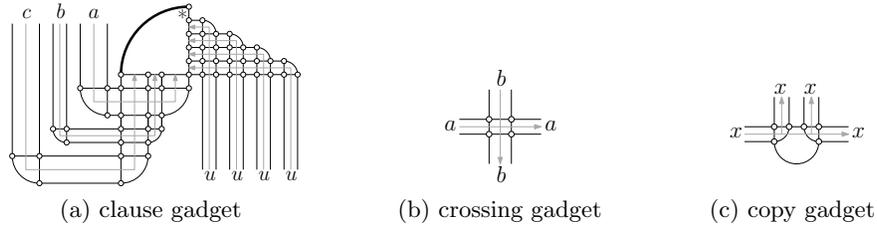

	\centering
	\subfloat[\label{fig:clauseGadget}{clause gadget}]{
	\includegraphics[scale=0.32,page=4]{smogNPFigures}}
	\hfil
	\subfloat[\label{fig:crossingGadget}{crossing gadget}]{
	\includegraphics[scale=0.35,page=8]{smogNPFigures}}
	\hfil
	\subfloat[\label{fig:copyGadget}{copy gadget}]{
	\includegraphics[scale=0.35,page=9]{smogNPFigures}}
	\caption{Different gadgets; gray-colored arrows show the information ``flow''.}
	\label{fig:gadgets}
\end{figure}

\algostep{Auxiliary gadgets} The \emph{crossing gadget} just consists of a rectangle and is used to allow two flows of information to cross each other; see Fig.~\ref{fig:crossingGadget}. The \emph{copy gadget} takes an information and creates three copies of this information; see Fig.~\ref{fig:copyGadget}. This is because both quarter circular arcs of the copy gadget must have the same radius in the presence of the half circular arc of the copy gadget. Finally, the \emph{unit length gadget} is a single edge, which we assume to be of length $\ell(u)$.

\smallskip We now describe our construction; see Fig.~\ref{fig:example}: $G_\varphi$ contains one unit length gadget, which is copied several times using the copy gadget (the number of copies depends linearly on the number of variables $\nu$ and clauses $\mu$ of $\varphi$). For each variable of $\varphi$, $G_\varphi$ has a variable gadget and a parity gadget, each of~which is connected to different copies of the unit length gadget. For each clause of~$\varphi$, $G_\varphi$ has a clause gadget, which has four connections to different copies of the unit length gadget. We compute $\mathcal{R}_\varphi$ as follows. We place the variable~gadget of each variable $x$ above and to the left of its parity gadget and we connect the output literals of the variable gadget of $x$ with its parity gadget through a copy gadget. We place the variable and the parity gadgets of the $i$-th variable below and to the right of the corresponding ones of the $(i-1)$-th variable. 
We place each clause gadget to the right of the sketch constructed so far, so that the gadget of the $i$-th clause is to the right of the $(i-1)$-th clause. 
This allows us to connect copies of the output literals of the variable gadget of each variable with the clause gadgets that contain it, so that all possible crossings (which are resolved using the crossing gadget) appear above the clause gadgets. More precisely, if a clause contains a literal of the $i$-th variable, we have a crossing with the literals of all variables with indices $(i+1)$ to $\nu$. Hence, for each clause we add $O(\nu)$ crossing and three copy gadgets.  Note that all copy gadgets of the unit length gadget lie below all variable, parity, and clause gadgets. The obtained~representation $\mathcal{R}_\varphi$ conforms with the one of Fig.~\ref{fig:example}. The construction can be done in $O(\nu  \mu)$~time.

To complete the proof, assume that $G_\varphi$ admits a bendless smooth orthogonal drawing $\Gamma_\varphi$ preserving $\mathcal{R}_\varphi$.
For each variable $x$ of $\varphi$, we set $x$ to $\texttt{true}$ if and only if $\ell(x) \geq 1.933 \cdot \ell(u)$. Since for each clause~$(a \vee b \vee c)$ of $\varphi$ we have that $\ell(a) + \ell(b) + \ell(c) > 4 \cdot \ell(u)$, at least one of $a$, $b$ and $c$ must be $\texttt{true}$. Hence, $\varphi$ admits a truth assignment. For the opposite direction, based on a truth assignment of $\varphi$, we can set, e.g., $\ell(x)=1.95$ and $\ell(\overline{x})=1.05$ for each variable $x$, assuming that $\ell(u)=1$. Then, arranging the variable and the clause gadgets of $G_\varphi$ as in Fig.~\ref{fig:example} yields a bendless smooth orthogonal drawing $\Gamma_\varphi$ preserving $\mathcal{R}_\varphi$.\qed
\end{proof}

\begin{remark}
The special case of our problem, in which circular arcs are not present, is known as \emph{HV-rectilinear planarity testing}~\cite{DBLP:conf/gd/ManuchPPT10}. As opposed to our problem, HV-rectilinear planarity testing is polynomial-time solvable in the fixed embedding setting~\cite{DBLP:conf/latin/DurocherF0M14} (and becomes $\mathcal{NP}$-hard in the variable embedding setting~\cite{DBLP:conf/gd/DidimoLP14}).
\end{remark}

\begin{theorem}\label{thm:np-octilinear}
Given a planar graph $G$ of max-degree~$4$ and an octilinear representation $\mathcal{R}$, it is $\mathcal{NP}$-hard to decide whether $G$ admits a bendless octilinear drawing preserving $\mathcal{R}$.
\end{theorem}
\begin{proof}[sketch]
Except for the parity gadget, we can adjust to the octilinear model simply by replacing arcs with diagonal segments; for details see \arxapp{\cite{arxiv}}{App.~\ref{app:preliminaries}}. In this case the parity gadget guarantees $|\ell(x)-\ell(\overline{x})| > {5}/{6} \cdot \ell(u) \approx 0.833 \cdot \ell(u)$, which implies that $\varepsilon < 0.084 \cdot \ell(u)$.\qed
\end{proof}

\section{Bi-Monotone Drawings}
\label{sec:kandinsky}

In this section, we study variants of the \emph{Kandinsky} drawing model~\cite{kandinsky-2,kandinsky-3,kandinsky}, which forms an extension of the orthogonal model to graphs of degree greater than $4$. In this model, the vertices are represented as squares, placed on a \emph{coarse grid}, with multiple edges attached to each side of them aligned on a \emph{finer grid}.

The Kandinsky model allows for natural extensions to both smooth orthogonal and octilinear models. We are aware of only one preliminary result in this~direction: 
A linear time drawing algorithm is presented in~\cite{smog1} for the production of smooth orthogonal $2$-drawings for planar graphs of arbitrary degree in quadratic area, in which all vertices are on a line $\ell$ and the edges are drawn either as half circles (above or below~$\ell$), or as two consecutive half~circles one above and one below $\ell$ (i.e., the latter ones are of complexity~$2$, but they are at most $n-2$).

For an input maximal planar graph $G$ (of arbitrary degree), our goal is to~construct a smooth orthogonal (or an octilinear) $2$-drawing for $G$ with the following aesthetic benefits over the aforementioned drawing algorithm: %
\begin{inparaenum}[(i)]
\item the vertices are distributed evenly over the drawing area, and
\item each edge is \emph{bi-monotone}~\cite{DBLP:conf/wg/FulekPSS11}, i.e., $xy$-monotone.
\end{inparaenum}
We achieve our goal at the cost of~slightly more edges drawn with complexity~$2$ or at the cost of increased drawing area (but still polynomial).

Our first approach is a modification of the \emph{shift-method}~\cite{dFPP}. Based on a canonical order $\pi =(v_1,\ldots,v_n)$ of $G$, we construct a planar smooth orthogonal $2$-drawing $\Gamma$ of $G$ in the Kandinsky model, as follows. We place $v_1$, $v_2$ and $v_3$ at $(0,0)$, $(2,0)$ and $(1,1)$. Hence, we can draw $(v_1,v_2)$ as a horizontal segment, and each of $(v_1,v_3)$ and $(v_2,v_3)$ as a quarter circular~arc. We also color $(v_1,v_3)$ blue and $(v_2,v_3)$ green.
For $k=4,\ldots,n$, assume that a smooth orthogonal $2$-drawing $\Gamma_{k-1}$ of the subgraph $G_{k-1}$ of $G$ induced by $v_1,\ldots,v_{k-1}$ has been constructed, in which each edge of the outerface $C_{k-1}$ of $\Gamma_{k-1}$ is drawn as a quarter circular arc, whose endvertices are on a line with slope $\pm 1$, except for edge $(v_1,v_2)$, which is drawn as a horizontal segment (called \emph{contour condition} in the shift-method; see Fig.~\ref{fig:shift-smooth}). Each of $v_1,\ldots,v_{k-1}$ is also associated with a so-called \emph{shift-set}, which for $v_1$, $v_2$ and $v_3$ are singletons containing only themselves.

Let $w_1,\ldots,w_p$ be the vertices of $C_{k-1}$ from left to right in $\Gamma_{k-1}$, where $w_1=v_1$ and $w_p=v_2$. Let $(w_\ell,\ldots,w_r)$, $1 \leq \ell < r \leq p$, be the neighbors of $v_k$ from left to right along $C_{k-1}$ in $\Gamma_{k-1}$. As in the shift-method, our algorithm first translates each vertex in $\cup_{i=1}^{\ell} S(w_i)$ one unit to the left and each vertex in $\cup_{i=r}^{p} S(w_i)$ one unit to the right, where $S(v)$ is the shift-set of $v \in V$. During this translation, $(w_{\ell},w_{\ell+1})$ and $(w_{r-1},w_r)$ acquire a horizontal segment each (see the bold edges of Fig.~\ref{fig:shift-smooth}). We place $v_k$ at the intersection of line $L_\ell$ with slope $+1$ through $w_\ell$ with line $L_r$ with slope $-1$ through $w_r$ (dotted in Fig.~\ref{fig:shift-smooth}) and we set the shift-set of $v_k$ to $\{v_k\} \cup_{i=\ell+1}^{r-1}S(w_i)$, as in the shift-method. We draw each of $(w_\ell,v_k)$ and $(v_k,w_r)$ as a quarter circular arc.
For $i=\ell+1,\ldots,r-1$, $(w_i,v_k)$ has a vertical line-segment that starts from $w_i$ and ends either at $L_\ell$ or $L_r$ and a quarter circle arc from the end of the previous segment to $v_k$. Hence, the contour condition is satisfied. We color $(w_\ell,v_k)$ blue, $(v_k,w_r)$ green and the remaining edges of $v_k$ red; see also~\cite{felsner,Schnyder90}. Observe that each blue and green edge~consists of a quarter circular arc and a horizontal segment (that may have zero~length), while a red edge consists of a vertical segment and a quarter circular arc (that may have zero radius). We are now ready to state our first theorem; the analogous of Theorem~\ref{theo:niceSC2} for the octilinear model is shown in \arxapp{\cite{arxiv}}{App.~\ref{app:kandinsky}}.

\begin{figure}[t!]
	\centering
	\includegraphics[width=.47\textwidth,page=3]{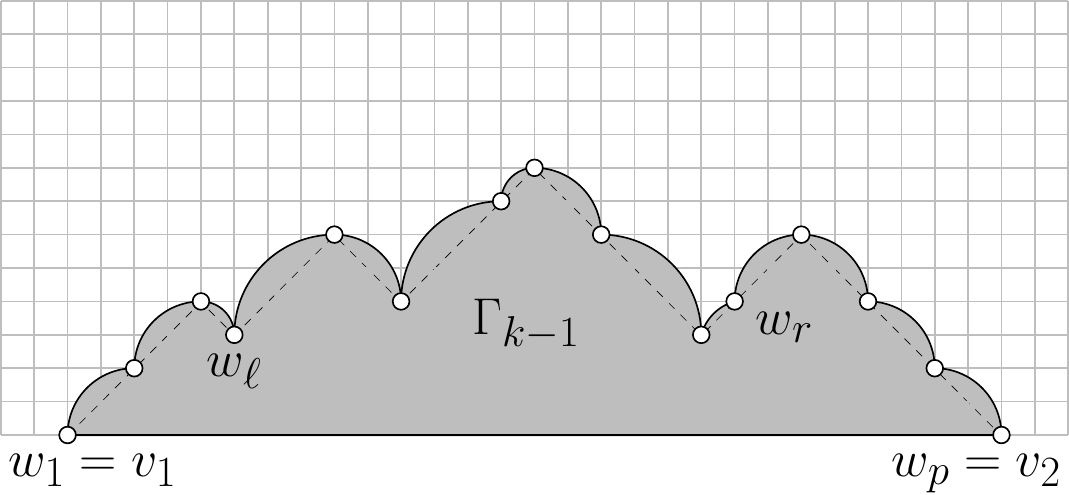}
	\hfill
	\includegraphics[width=.47\textwidth,page=4]{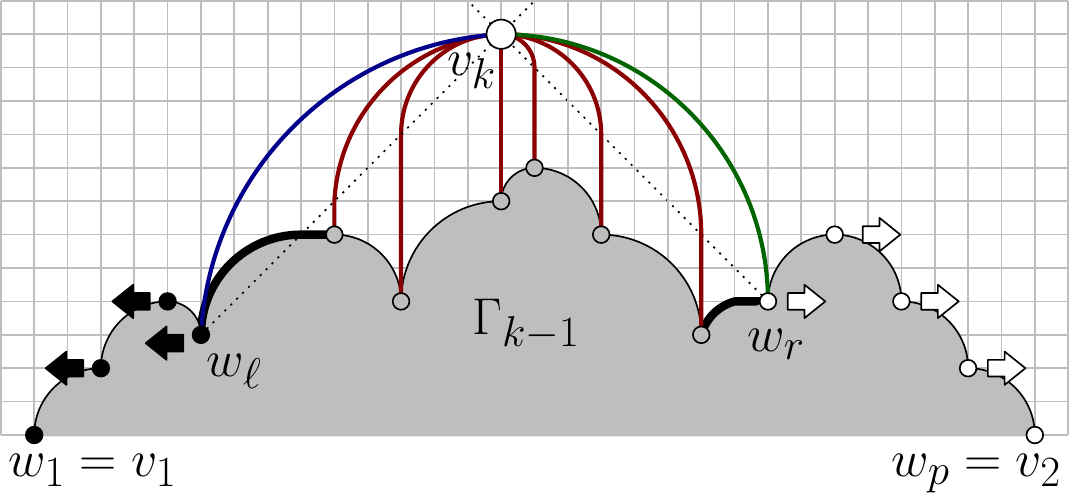}
	\caption{%
	Illustration of the contour condition (left) and the placement of $v_k$ (right).}
	\label{fig:shift-smooth}
\end{figure}

\begin{theorem}
A maximal planar $n$-vertex graph admits a bi-monotone planar smooth orthogonal $2$-drawing in the Kandinsky model, which requires $O(n^2)$ area and can be computed in $O(n)$ time.
\label{theo:niceSC2}
\end{theorem}
\begin{proof}
Bi-monotonicity follows by construction.
The time complexity follows from~\cite{DBLP:journals/ipl/ChrobakP95}. Planarity is proven by induction. Drawing $\Gamma_3$ is planar by construction. Assuming that $\Gamma_{k-1}$ is planar, we observe that no two edges incident to $v_k$ cross in $\Gamma_k$. Also, these edges do not cross edges of $\Gamma_{k-1}$. Since the radii of the arcs of the edges incident to vertices that are shifted remain unchanged and since edges incident to vertices in the shift-sets retain their shape, drawing $\Gamma_k$ is planar.\qed
\end{proof}

We reduce the number of edges drawn with complexity~$2$ in two steps.~%
\begin{inparaenum} [({S.}1)]
\item \label{s:h-stretch} We stretch the drawing horizontally (by employing appropriate vertical cuts; \arxapp{see, e.g.,~\cite{4M}}{refer to App.~\ref{app:preliminaries}}) to eliminate the vertical segments of all red edges with a circular arc segment of non-zero radius.
\item \label{s:v-stretch} We stretch the drawing vertically, to guarantee that the edges of a spanning tree (i.e., $n-1$) are drawn with complexity~$1$.
\end{inparaenum}

For Step~S.\ref{s:h-stretch}, we assume that each blue and green edge has a horizontal~segment (that may be of zero length). Consider a red edge $(u,v)$ with a vertical segment of length $\delta$ and assume w.l.o.g.\ that $u$ is to the right and above $v$\arxapp{}{; see Fig.~\ref{fig:h-stretch}}. If we shift $u$ by $\delta$ units to the right, then $(u,v)$ can be drawn  as a quarter circular arc. If the shift is by more than $\delta$ units, then a horizontal segment is needed. Since all  edges incident to $u$ that are drawn below $u$ enter $u$ from its left or from its right side, the shift of $u$ cannot introduce crossings between them.

We eliminate the vertical segments of all red edges with a circular arc segment of non-zero radius, as follows. As long as there exist such edges, we choose the one, call it $(u,v)$, whose vertical segment has the largest length $\delta$, and assume that $u$ is to the right and above $v$. We eliminate the vertical segment of $(u,v)$ using a vertical cut $L$ at $x(u)- \varepsilon$, for small $\varepsilon>0$. Since $L$ crosses several edges, shifting all vertices to the right of $L$ by $\delta$ to the right has the following effects. By the choice of~$(u,v)$, the vertical segments of all red edges crossed by $L$ are eliminated; note that this may introduce new horizontal segments. The horizontal  segment of each blue and green edge crossed by $L$ is elongated by $\delta$. Both imply that no edge crossings are introduced. Hence, by the termination of our algorithm all edges with vertical segments are of complexity~$1$.

Step~\ref{s:h-stretch} ensures that the $x$-distance of adjacent vertices is at least as large as their $y$-distance (unless they are connected by vertical edges). Based on this property, in Step~\ref{s:v-stretch} we compute new $y$-coordinates for the vertices in the sequence of the canonical ordering $\pi$, keeping their $x$-coordinates unchanged. First, we set $y(v_1) = y(v_2) = 0$. For each $k=3,\ldots,n$, we set $y(v_k) = \max_{w\in\{w_\ell,\ldots,w_r\}}\{y(w) + \max\{\Delta_x(v_k, w),1\} \}$, where $w_\ell,\ldots,w_r$ are the neighbors of $v_k$ in $\Gamma_{k-1}$, i.e., $v_k$ is placed above $w_\ell,\ldots,w_r$ in $\Gamma_{k-1}$, such that one of its edges (the one of the maximum; call it $(v_k,w^\ast)$) is drawn with complexity~$1$; as a quarter circle arc or as a vertical edge depending on whether the $x$- distance of $v_k$ and $w^\ast$ is non-zero or not. Since $(v_k,w^\ast)$ is the edge that must be stretched the most in order to ensure that it is drawn with complexity~$1$, for all other edges incident to $v_k$ in $G_k$, the $y$-distance of their endpoints is at least as large as their corresponding $x$-distance. Hence, they are drawn as vertical segments followed by quarter circular arcs (that may have zero radius). We are now ready to state our second theorem.

\newcommand{\evenNicerSCTwo}{A maximal planar $n$-vertex graph admits a bi-monotone planar smooth orthogonal $2$-drawing with at least $n-1$ edges with complexity~$1$ in the Kandinsky model, which requires $O(n^4)$ area and can be computed in $O(n^2)$ time.}
\begin{theorem}
\evenNicerSCTwo
\label{theo:evenNicerSC2}
\end{theorem}
\begin{proof}[sketch]
For $k=3,\ldots,n$, vertex $v_k$ is incident to an edge drawn with complexity~$1$ in Step~\ref{s:v-stretch}. Since $(v_1,v_2)$ is drawn as a horizontal segment, at least $n-1$ edges have complexity~$1$. Planarity is proven by induction; the main invariant is that all edges on $C_k \setminus \{(v_1,v_2)\}$ have a quarter circular arc and possibly a vertical segment. Time and area requirements are shown in \arxapp{\cite{arxiv}}{App.~\ref{app:kandinsky}}.\qed
\end{proof}

\section{Conclusions}
\label{sec:conclusions}

In this paper, we continued the study on smooth orthogonal and octilinear drawings. Our $\mathcal{NP}$-hardness proofs are a first step towards settling the complexity~of both drawing problems. We conjecture that the former is $\mathcal{NP}$-hard, even in the case where only the planar embedding is specified by the input. For the latter, it is of interest to know if it remains $\mathcal{NP}$-hard even for planar graphs of max-degree~$4$ or if these graphs allow for a decision algorithm. Our drawing algorithms guarantee bi-monotone $2$-drawings with a certain number of complexity-$1$ edges for maximal planar graphs. Improvements on this number or generalizations to triconnected or simply connected planar graphs are of~importance.

\paragraph{Acknowledgements.} The authors would like to thank Patrizio Angelini and Martin Gronemann for useful discussions.

\bibliographystyle{splncs03}
\bibliography{references}
\arxapp{}{\newpage
\appendix
\section*{Appendix}

\section{Preliminary Notions and Definitions}
\label{app:preliminaries}

Unless otherwise specified, we consider simple undirected graphs. Let $G=(V,E)$ be a graph. We denote by $n$ and $m$ the number of vertices and edges of $G$.
We say that $G$ has \emph{max-degree} $\Delta$, if $G$ has no vertex with degree larger than $\Delta$. 

A planar drawing $\Gamma$ of $G$ partitions the plane into connected regions, called \emph{faces}; the unbounded one is called \emph{outerface}. A (\emph{topological}) \emph{planar embedding}~$\mathcal{E}$ of $G$ is an equivalence class of planar drawings that define the same set of faces and outerface. Embedding $\mathcal{E}$ can equivalently be defined by the cyclic orders of the edges incident to each vertex (also called \emph{combinatorial~embedding}). Given a drawing $\Gamma$, a \emph{vertical cut} is a vertical line, which crosses only horizontal edge segments of the drawing and splits it into two parts; a left and right one~\cite{4M}. If one shifts the right part to the right, while keeping the left part in place, the result has no crossings. A \emph{horizontal cut} is defined analogously.

\begin{figure}[b!]
	\centering
	\begin{minipage}[b]{0.48\textwidth}
		\centering
		\subfloat[\label{fig:shift-method-1} {contour condition}]{
		\includegraphics[width=\textwidth,page=1]{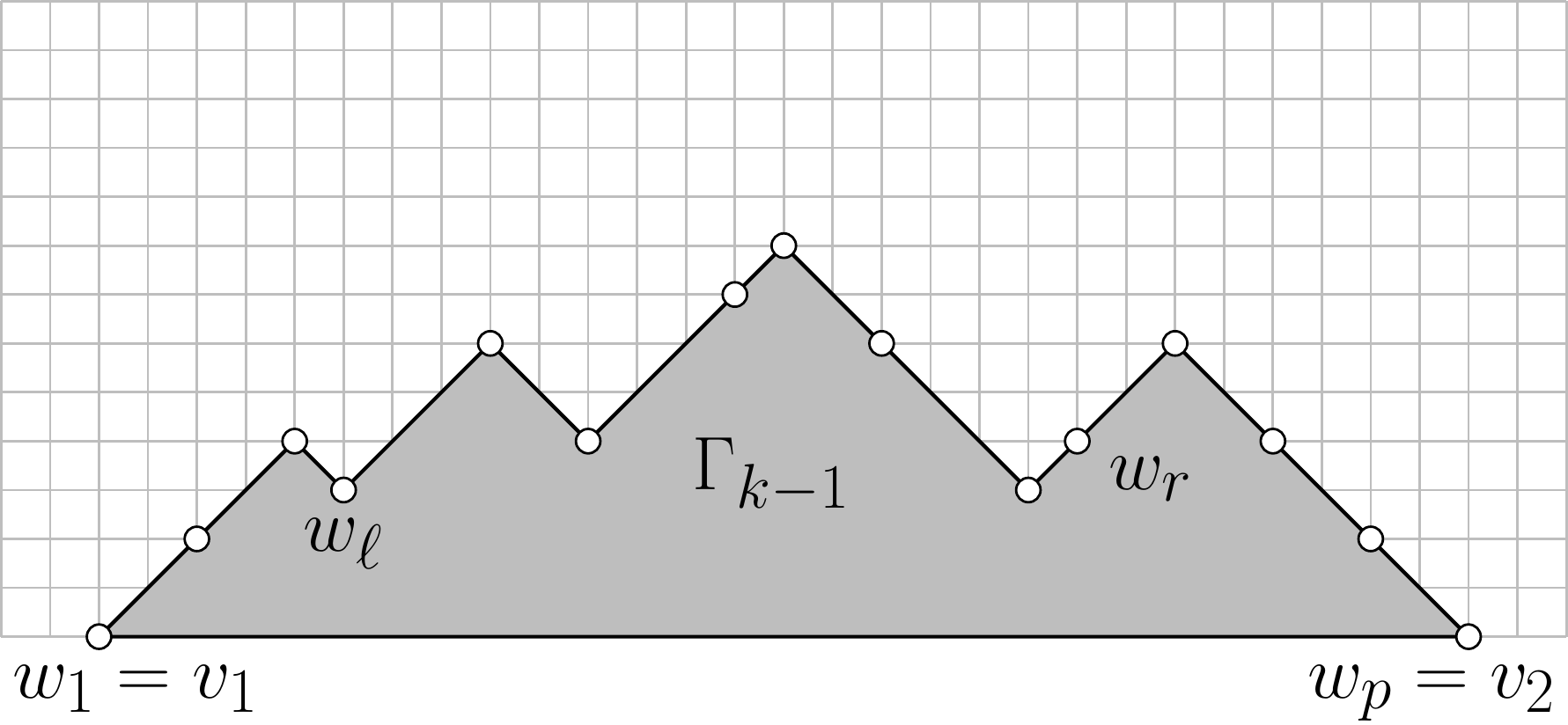}}
	\end{minipage}
	\hfill
	\begin{minipage}[b]{0.48\textwidth}
		\centering
		\subfloat[\label{fig:shift-method-2} {placement of $v_k$ in $\Gamma_{k-1}$}]{
		\includegraphics[width=\textwidth,page=2]{canonical}}
	\end{minipage}
	\caption{%
	Illustration of the shift-method by de~Fraysseix, Pach and Pollack~\cite{dFPP}.}
	\label{fig:shift-method}
\end{figure}

The \emph{canonical order} for maximal planar graphs~\cite{dFPP} is formally defined as follows. Let $G=(V,E)$ be a maximal planar graph and let $\pi =(v_1,\ldots,v_n)$ be a permutation of $V$. Assume that edges $(v_1,v_2)$, $(v_2,v_n)$ and $(v_1,v_n)$ form the outerface of $G$. For $k=1,\ldots,n$, let $G_k$ be the subgraph induced by $\cup_{i=1}^k v_i$ and denote by $C_k$  the outerface of $G_k$. Then, $\pi$ is a \emph{canonical ordering} of $G$ if for each $k=2,\ldots,n$ the following hold: %
\begin{inparaenum}[(i)]
\item $G_k$ is biconnected,
\item all neighbors of $P_k$ in $G_{k-1}$ are on $C_{k-1}$, and
\item all vertices of $P_k$ with $2\leq k < n$ have at least one neighbor in $P_j$ for some $j > k$.
\end{inparaenum}
A canonical ordering of a maximal planar graph can be computed in linear time~\cite{cannonicalorder}.

The \emph{shift-method}~\cite{dFPP} is a well-known linear-time algorithm, which constructs a planar drawing $\Gamma$ of a maximal planar graph $G=(V,E)$ on a grid of quadratic area, based on a canonical order $\pi$ of $G$ as follows. It places $v_1$, $v_2$ and $v_3$ at points $(0,0)$, $(2,0)$ and $(1,1)$. For $k=4,\ldots,n$, assume that a planar drawing $\Gamma_{k-1}$ of $G_{k-1}$ has been constructed in which each edge of $C_{k-1}$ is drawn as a straight-line segment with slope $\pm 1$, except for the edge $(v_1,v_2)$, which is drawn as a horizontal line segment (\emph{contour condition}; see Fig.~\ref{fig:shift-method-1}) and that each of the vertices $v_1,\ldots,v_{k-1}$ has been associated with a so-called \emph{shift-set}, which for $v_1$, $v_2$ and $v_3$ are singletons containing only themselves. Let $w_1,\ldots,w_p$ be the vertices of $C_{k-1}$ from left to right in $\Gamma_{k-1}$, where $w_1=v_1$ and $w_p=v_2$. Let also $(w_\ell,\ldots,w_r)$, with $1 \leq \ell < r \leq p$ be the neighbors of $v_k$ from left to right along $C_{k-1}$ in $\Gamma_{k-1}$. To avoid edge-overlaps, the algorithm first translates each vertex in $\cup_{i=1}^{\ell} S(w_i)$ one unit to the left and each vertex in $\cup_{i=r}^{p} S(w_i)$ one unit to the right, where $S(v)$ denotes the shift-set of $v \in V$. Then, the algorithm places $v_k$ at the intersection of the line with slope $+1$ through $w_\ell$ with the line with slope $-1$ through $w_r$ and sets the shift-set of $v_k$ to $\{v_k\} \cup_{i=\ell+1}^{r-1}S(w_i)$; see Fig.~\ref{fig:shift-method-2}.

\section{Omitted Proofs from Section~\ref{sec:relationships}}
\label{app:relationships}

\rephrase{Theorem}{\ref{theo:union}}{\union}
\begin{proof}
For each $k \in \mathbb{N}_+$ we describe a $4$-regular planar graph $G_k=(V_k,E_k)$ with $20k$ vertices that admits both a bendless smooth orthogonal drawing and a bendless octilinear drawing; refer to Fig.~\ref{fig:trains} for the case $k=2$. $G_k$ has $4k$ subgraphs $W_{i,j}$ such that $1 \leq i \leq 2k$ and $j \in \{t,b\}$ (top and bottom). Graph $W_{i,j}$ consists of five vertices $c_{i,j}$, $n_{i,j}$, $w_{i,j}$, $e_{i,j}$, and $s_{i,j}$ (center, north, west, east, south, respectively), such that $W_{i,j}$ is a wheel on five vertices, i.e., $W_{i,j}$ consists of a center-vertex $c_{i,j}$ and a cycle $C_{i,j}=(n_{i,j},w_{i,j},s_{i,j},e_{i,j})$, such that $c_{i,j}$ is connected to all vertices of cycle $C_{i,j}$.

All vertices except for $c_{i,j}$ already have degree three in $W_{i,j}$. So, we just have to describe the missing edges that make $G_k$ $4$-regular: For $1 \leq h \leq 2k-1$, $(e_{h,j},w_{h+1,j}) \in E_k$ for $j \in \{t,b\}$; blue edges in Fig.~\ref{fig:trains}. Also, $(w_{1,t},w_{1,b}) \in E_k$ and $(e_{2k,t},e_{2k,b}) \in E_k$; gray edges in Fig.~\ref{fig:trains}. For $1 \leq h \leq 2k$, $(s_{h,t},n_{h,b}) \in E_k$; red edges in Fig.~\ref{fig:trains}. Finally, for $1 \leq h \leq k$, $(n_{2h-1,t},n_{2h,t}) \in E_k$ and $(s_{2h-1,b},s_{2h,b}) \in E_k$; green edges in Fig.~\ref{fig:trains}. With those additional edges, $G_k$ becomes $4$-regular. Fig.~\ref{fig:trains} is a certificate that $G_k=(V_k,E_k)$ indeed admits both a bendless smooth orthogonal drawing and a bendless octilinear drawing.\qed
\end{proof}

\begin{lemma}
The graph $B$ of Fig.~\ref{fig:necklaceGem} does not admit a bendless smooth orthogonal drawing, when the outerface is fixed to $(p_1,q_1,p_2,q_2)$, and each of $q_1$ and $q_2$ have two unoccupied ports on the outerface.
\label{lem:necklaceGemNotOcti}
\end{lemma}
\begin{proof}
First, we discuss some structural properties of graph $B$. Observe~that graph $B$ contains a wheel $W_5$ on five vertices as a subgraph, which is induced by the vertices drawn as circles in Fig.~\ref{fig:necklaceGem}. Its center is vertex $c$ and its rim consists of vertices $w_1$, $w_2$, $w_3$, and $w_4$. Vertices $w_1$ and $w_2$ form a triangular face with vertex $t_1$; vertices $w_3$ and $w_4$ form a triangular face with $t_2$ (vertices $t_1$ and $t_2$ are drawn as triangles in Fig.~\ref{fig:necklaceGem}). Observe that $t_1$ and $t_2$ form a separation pair and both are connected to vertices $p_1$ and $p_2$ (drawn as pentagons in Fig.~\ref{fig:necklaceGem}) forming two pentagonal faces $(p_1,t_1,w_1,w_4,t_2)$ and $(p_2, t_2, w_3, w_2, t_1)$. Observe that $p_1$ and $p_2$ also form a separation pair and are both connected to vertices $q_1$ and $q_2$ (drawn as squares in Fig.~\ref{fig:necklaceGem}) forming two quadrilateral faces $(q_1, p_2, t_1, p_1)$ and $(q_2,p_1, t_2, p_2)$. Hence, $B$ has two separation pairs and two vertices of degree~$2$ (i.e., $q_1$ and $q_2$). The remaining vertices of $B$ are of degree exactly~$4$.

In order to show that $B$ does not admit a bendless smooth orthogonal drawing, when the outerface is $(p_1,q_1,p_2,q_2)$, and each of $q_1$ and $q_2$ have two unoccupied ports on the outerface, we first observe the following: If we want to draw $W_5$ such that all of its unoccupied ports are on its outerface, then none of its four triangular faces must have an unoccupied port pointing in its interior.  In the bendless smooth orthogonal model, there are only two possible drawings for a triangular face fulfilling this property, as shown in~\cite{smog2}; see Figs.~\ref{fig:sc1triangle1} and~\ref{fig:sc1triangle2}. This implies that $W_5$ admits only two bendless smooth orthogonal drawings such that all of its unoccupied ports are on its outerface; see Figs.~\ref{fig:sc1wheel1} and~\ref{fig:sc1wheel2}.

\begin{figure}[t]
	\centering
	\begin{minipage}[b]{0.225\textwidth}
		\centering
		\subfloat[\label{fig:sc1triangle1} {}]{
		\includegraphics[scale=0.5,page=10]{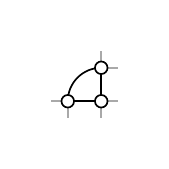}}
	\end{minipage}
	\begin{minipage}[b]{0.225\textwidth}
		\centering
		\subfloat[\label{fig:sc1triangle2} {}]{
		\includegraphics[scale=0.5,page=11]{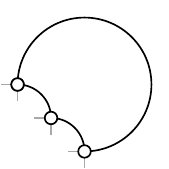}}
	\end{minipage}
	\begin{minipage}[b]{0.225\textwidth}
		\centering
		\subfloat[\label{fig:sc1wheel1} {}]{
		\includegraphics[scale=0.5,page=12]{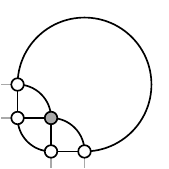}}
	\end{minipage}
	\begin{minipage}[b]{0.225\textwidth}
		\centering
		\subfloat[\label{fig:sc1wheel2} {}]{
		\includegraphics[scale=0.5,page=13]{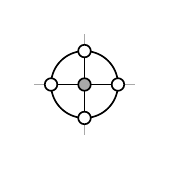}}
	\end{minipage}
	\caption{%
	All possible drawings of (a)-(b)~a triangular face, and of (c)-(d)~a wheel on five vertices, 
	such that all unoccupied ports are on the outerface.}
	\label{fig:sc1drawings} 
\end{figure}

Next, we consider $t_1$ and $t_2$.  Since $t_1$ and $t_2$ define two triangular faces in the subgraph induced by $W_5$, and $t_1$ and $t_2$, similar as above, we can conclude that there are five different drawings of this graph; see  Fig.~\ref{fig:bluePart}. Note that $t_1$ and $t_2$ can be moved along diagonals of slope $1$ in
Figs.~\ref{fig:bluePart4} and~\ref{fig:bluePart5}.

\begin{figure}[bp]
	\centering
	\subfloat[\label{fig:bluePart1} {}]{
	\includegraphics[scale=0.5,page=14]{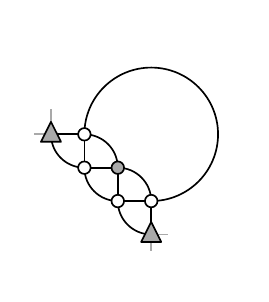}}
	\hfil
	\subfloat[\label{fig:bluePart2} {}]{
	\includegraphics[scale=0.5,page=15]{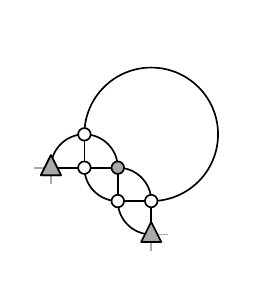}}
	\hfil
	\subfloat[\label{fig:bluePart3} {}]{
	\includegraphics[scale=0.5,page=16]{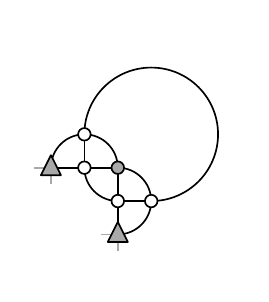}}
	\hfil
	\subfloat[\label{fig:bluePart4} {}]{
	\includegraphics[scale=0.5,page=17]{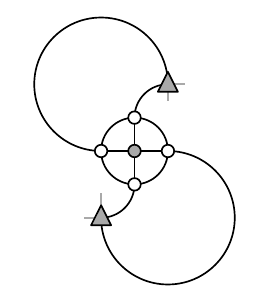}}
	\hfil
	\subfloat[\label{fig:bluePart5} {}]{
	\includegraphics[scale=0.5,page=18]{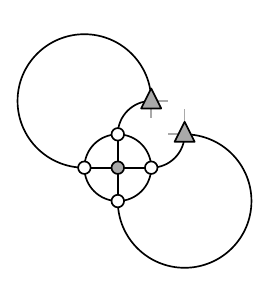}}
	\caption{All possible drawings of the subgraph induced by $W_5$, $t_1$ and $t_2$, such that all unoccupied ports are on the outerface.}
	\label{fig:bluePart}
\end{figure}

In the following step, we will consider all candidate positions for placing $p_1$ and $p_2$, which we can identify as follows: In a bendless smooth orthogonal drawing, both endpoints of an edge are located along a horizontal, vertical or diagonal line. Both $p_1$ and $p_2$ are neighbors of both $t_1$ and $t_2$, for which we already defined their locations. If we consider all rays emanating from $t_1$ and $t_2$ with slopes $\{0,1,-1,\infty\}$, then $p_1$ and $p_2$ must be located at an intersection of a ray emanating from $t_1$ and a ray emanating from $t_2$; see Fig.~\ref{fig:greenPartMethod1}. For each candidate position, we then try to draw the edges to $t_1$ and $t_2$ using one of the edge segments supported by the smooth orthogonal model. The resulting drawing is \emph{valid} if and only if none of the following arises: %
\begin{inparaenum}[(i)] 
\item a vertex has an unoccupied port not incident to the outerface; see  Fig.~\ref{fig:greenPartMethod2},
\item a port is used twice; see Fig.~\ref{fig:greenPartMethod3},
\item an edge is not drawn planar; see Fig.~\ref{fig:greenPartMethod4}.
\end{inparaenum} 
Recall that for the cases shown in Figs.~\ref{fig:bluePart4} and~\ref{fig:bluePart5}, we have to take into account all relevant combinations of positions of $t_1$ and $t_2$ along the diagonals.

\begin{figure}[t]
	\centering
	\subfloat[\label{fig:greenPartMethod1} {}]{
	\includegraphics[scale=0.425,page=19]{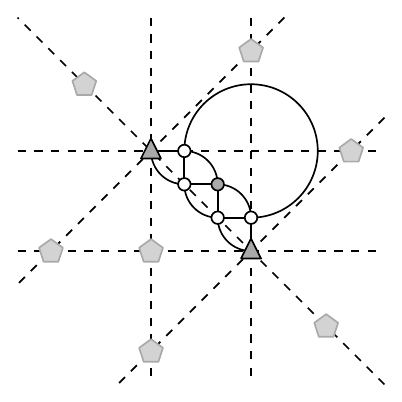}}
	\hfil
	\subfloat[\label{fig:greenPartMethod2} {}]{
	\includegraphics[scale=0.425,page=20]{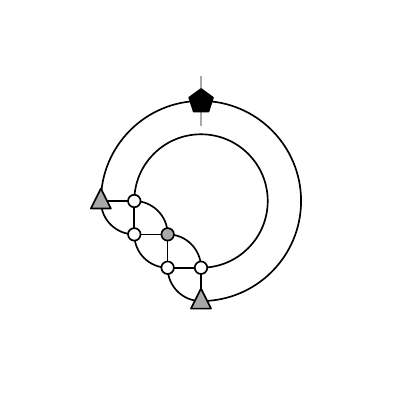}}
	\hfil
	\subfloat[\label{fig:greenPartMethod3} {}]{
	\includegraphics[scale=0.425,page=21]{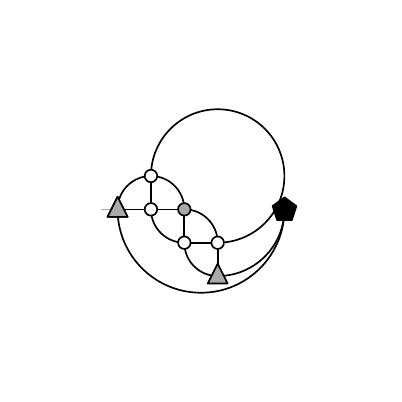}}
	\hfil
	\subfloat[\label{fig:greenPartMethod4} {}]{
	\includegraphics[scale=0.425,page=22]{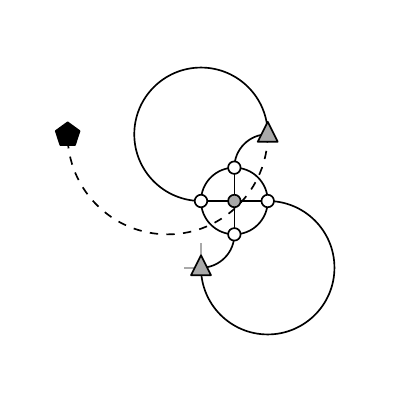}}
	\caption{Method used for identifying valid drawings for $p_1$ and $p_2$:
	(a)~Identification of candidate positions. 
	(b)-(d)~Cases of invalid drawings.}
	\label{fig:greenPartMethod}
\end{figure}

As a result of our analysis, we can conclude that the only valid drawings of the subgraph induced by $W_5$, $t_1$, $t_2$ and at least one of $p_1$ and $p_2$ are those shown in Fig.~\ref{fig:greenPart}. Note that in the cases shown in Figs.~\ref{fig:greenPart2}-\ref{fig:greenPart7}, we can only place one of $p_1$ and $p_2$. For the case shown in Fig.~\ref{fig:greenPart1} we proceed by considering all candidate positions of $q_1$ and $q_2$, as we did for $p_1$ and $p_2$. As a result, we conclude that $q_1$ and $q_2$ cannot be added such that each of them has two unoccupied ports on the outerface, which completes the proof of this lemma.\qed
\end{proof}

\begin{figure}[!tp]
	\centering
	\begin{minipage}[b]{\textwidth}
		\centering
		\subfloat[\label{fig:greenPart1} {}]{
		\includegraphics[scale=0.35,page=26]{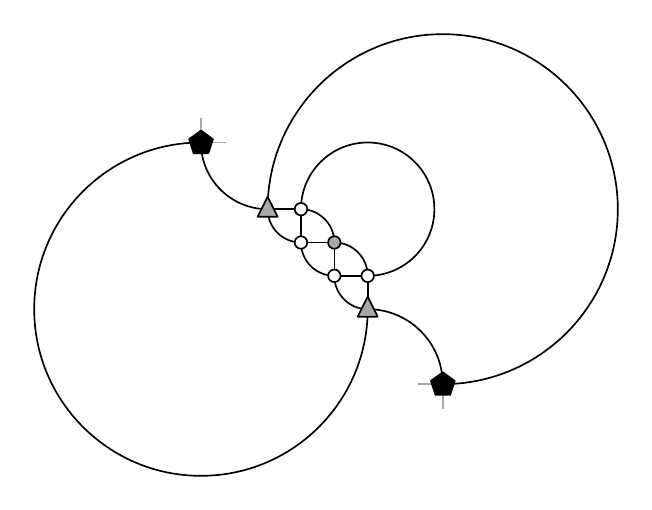}}
	\end{minipage}
	\subfloat[\label{fig:greenPart2} {}]{
	\includegraphics[scale=0.35,page=27]{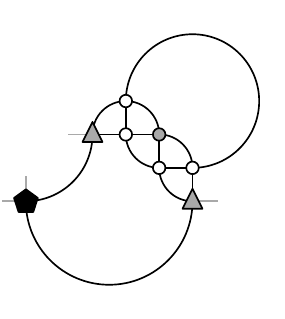}}
	\hfil
	\subfloat[\label{fig:greenPart3} {}]{
	\includegraphics[scale=0.35,page=28]{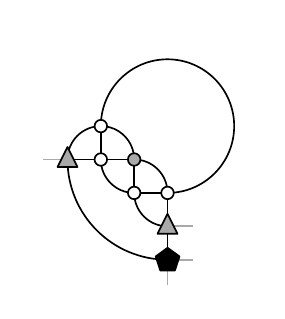}}
	\hfil
	\subfloat[\label{fig:greenPart4} {}]{
	\includegraphics[scale=0.35,page=29]{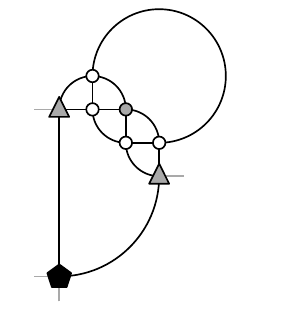}}
	\hfil
	\subfloat[\label{fig:greenPart5} {}]{
	\includegraphics[scale=0.35,page=30]{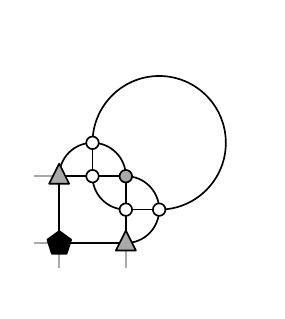}}
	\hfil
	\subfloat[\label{fig:greenPart6} {}]{
	\includegraphics[scale=0.35,page=31]{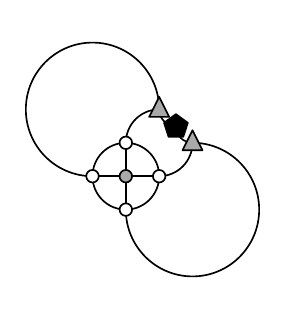}}
	\hfil
	\subfloat[\label{fig:greenPart7} {}]{
	\includegraphics[scale=0.35,page=32]{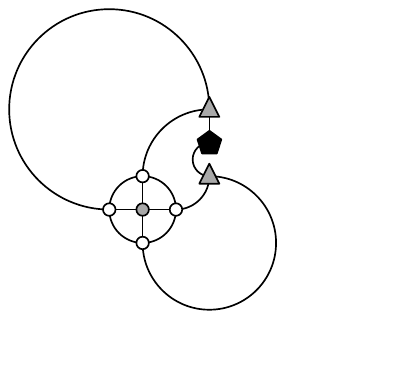}}
	\caption{
	All valid drawings of the subgraph induced by $W_5$, $t_1$, $t_2$, and at least one of $p_1$ and $p_2$.}
	\label{fig:greenPart}
\end{figure}

\FloatBarrier
\pagebreak

\section{Omitted Proofs and Material from Section~\ref{sec:nphardness}}
\label{app:np-hardness}

\begin{figure}[!h]
	\centering
	\includegraphics[scale=0.3,page=5]{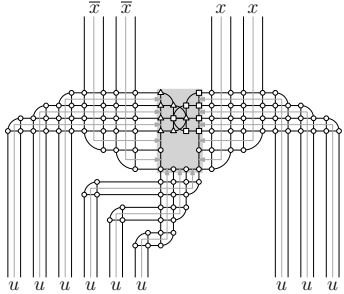}
	\caption{Illustration of the parity gadget when $\ell(x) \approx \ell(\overline{x})$.}
	\label{fig:parityGadgetInfeasible}
\end{figure}

\begin{lemma}\label{lem:partity-smooth}
For the smooth orthogonal model, the parity gadget of variable~$x$ of formula $\varphi$ contains no edge crossings if and only if $|\ell(\overline{x}) - \ell(x)|> \frac{\sqrt{3}}{2}\ell(u)$, where $\ell(u)$ denotes the unit length.
\end{lemma}
\begin{proof}
Refer to Fig.~\ref{fig:SC1ParityGadgetDetail}, which gives a more detailed
illustration of the vertical gap of the parity gadget. Consider the case where
$x = \texttt{false}$. The case where $x = \texttt{true}$ is symmetric. If $x =
\texttt{false}$, we want the bottom arc cut by the dashed diagonal to be
completely below the top arc cut by the dashed diagonal. Since we know that both
arcs have radius $\ell(u)$, their centers (gray-colored in
Fig.~\ref{fig:SC1ParityGadgetDetail}) should be a distance of $> 2 \cdot
\ell(u)$ apart from each other. This corresponds to the length of the dashed
diagonal. We can compute the length of the dashed diagonal in dependence of
$\lambda = \ell(\overline{x}) - \ell(x)$ as the dashed diagonal is part of a right triangle for which we know the remaining side lengths. Thus, we can use Pythagoras' theorem and compute the length of the diagonal which gives us $\lambda > \sqrt{3}/{2}\cdot \ell(u) \approx 0.866 \cdot \ell(u)$.\qed
\end{proof}

\begin{figure}[!h]
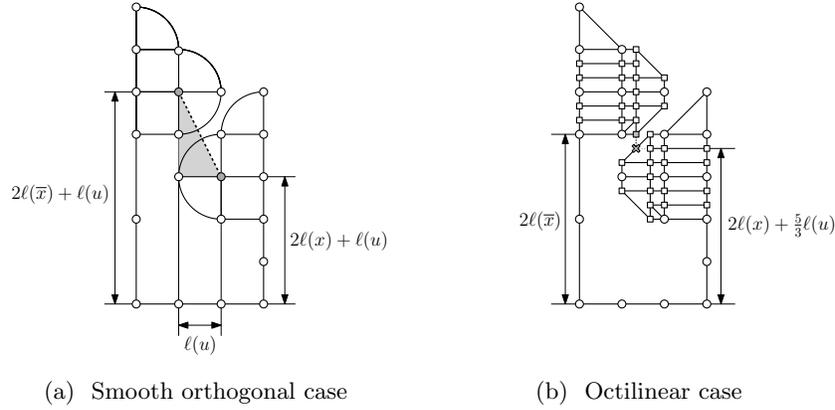

	\centering
	\subfloat[\label{fig:SC1ParityGadgetDetail} {Smooth orthogonal case}]{
	\includegraphics[scale=0.5,page=11]{smogNPFigures}}
	\hfil
	\subfloat[\label{fig:8C1ParityGadgetDetail} {Octilinear case}]{
	\includegraphics[scale=0.5,page=12]{smogNPFigures}}
	\caption{Illustrations for the proofs of Lemmas~\ref{lem:partity-smooth} and~\ref{lem:partity-octilinear}.}
	\label{fig:parityGadgetDetail}
\end{figure}

\noindent The corresponding parity gadget for the octilinear model is
illustrated in Fig.~\ref{fig:8C1ParityGadgetDetail}.

\begin{lemma}\label{lem:partity-octilinear}
For the octilinear model, the parity gadget of variable~$x$ of formula $\varphi$
contains no edge crossings if and only if $|\ell(\overline{x}) - \ell(x)|> \frac{5}{6}\ell(u)$, where $\ell(u)$ denotes the unit length.
\end{lemma}
\begin{proof}
Refer to Fig.~\ref{fig:8C1ParityGadgetDetail}. Consider the case where $x = \texttt{false}$. The case where $x = \texttt{true}$ is symmetric. In comparison to the smooth orthogonal setting we subdivide several axis-aligned edges of unit edge length $\ell(u)$ into three equally sized smaller edges (see the edges with endpoints drawn as squares in the figure) of length $\frac{1}{3}\ell(u)$. Note that we can realize this by actually composing each $\ell(u)$ in the representation of three edges of size $\frac{1}{3}\ell(u)$ instead. 

Since we do not deal with circular arc segments here, the smallest distance $d > 0$ between both blocks is easy to compute; in particular it is located between the bottom diagonal of the upper block and the middle diagonal of the lower block which are parallel. For instance, $d$ is the distance between the bottom endpoint of the bottom diagonal of the upper block $v_b$ (gray squared-shaped vertex in Fig.~\ref{fig:8C1ParityGadgetDetail}) and the point on the diagonal of the lower block that is vertically below $v_b$ (gray cross in Fig.~\ref{fig:8C1ParityGadgetDetail}) which is not a vertex). For both, we know the $y$-coordinates $2\cdot \ell(x)$ and $2\cdot \ell(\overline{x})+5/3 \cdot \ell(u)$, respectively. Hence, $d = 2 \cdot \left( \ell(\overline{x}) - \ell(x) \right) - 5/3 \cdot \ell(u) > 0$ which implies $\ell(\overline{x}) - \ell(x) > 5/6 \cdot \ell(u)$.\qed 
\end{proof}

\section{Omitted Proofs and Material from Section~\ref{sec:kandinsky}}
\label{app:kandinsky}

\begin{theorem}\label{theo:octiSC2}
A maximal planar $n$-vertex graph admits a bi-monotone planar octilinear $2$-drawing in the Kandinsky model, which requires $O(n^2)$ area and can be computed in $O(n)$ time. In addition, each bend is at~$135^\circ$.
\end{theorem}
\begin{proof}
\begin{figure}[h!]
	\centering
	\begin{minipage}[b]{0.48\textwidth}
		\centering
		\subfloat[\label{fig:octi-shift-method-1} {contour condition}]{
		\includegraphics[width=\textwidth,page=1]{canonical}}
	\end{minipage}
	\hfill
	\begin{minipage}[b]{0.48\textwidth}
		\centering
		\subfloat[\label{fig:octi-shift-method-2} {placement of $v_k$ in $\Gamma_{k-1}$}]{
		\includegraphics[width=\textwidth,page=9]{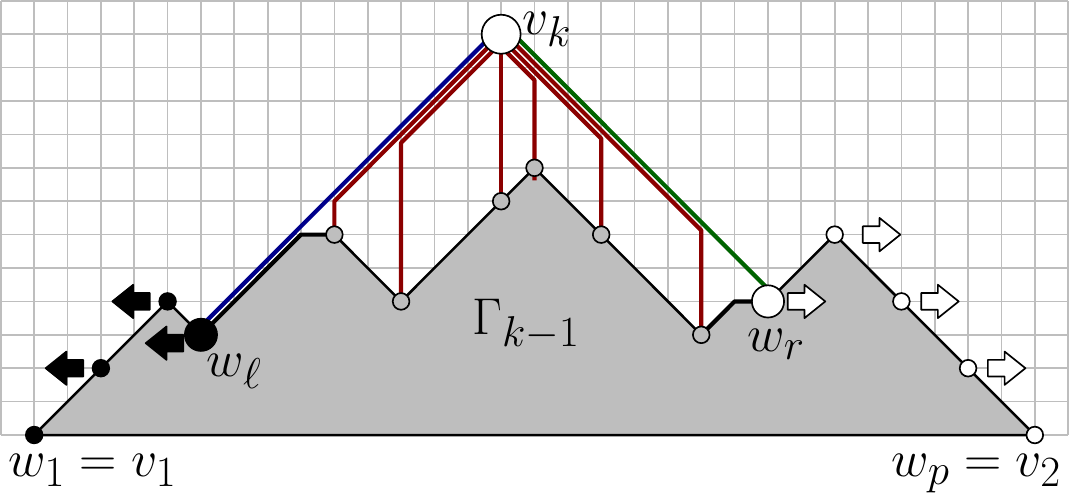}}
	\end{minipage}
	\caption{%
	Illustration of the modified shift-method for the octilinear Kandinsky model.}
	\label{fig:octi-shift-method}
\end{figure}
The proof is rather simple. We can actually convert the layout computed for the smooth orthogonal model to octilinear by redrawing all quarter circular arcs of it as diagonal segments; see also Fig.~\ref{fig:octi-shift-method-2}. This results in bends at~$135^\circ$. Planarity follows from blue and green edges not passing through vertices by virtue of construction.\qed
\end{proof}

\begin{figure}[h]
	\centering
	\includegraphics[width=.25\textwidth,page=5]{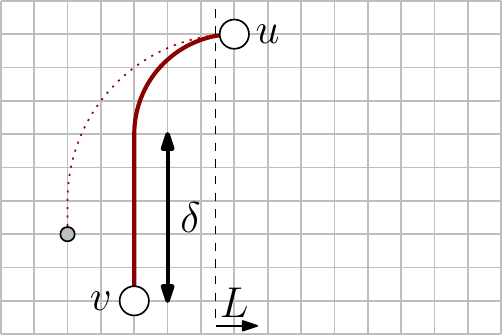}
	\hfil
	\includegraphics[width=.25\textwidth,page=6]{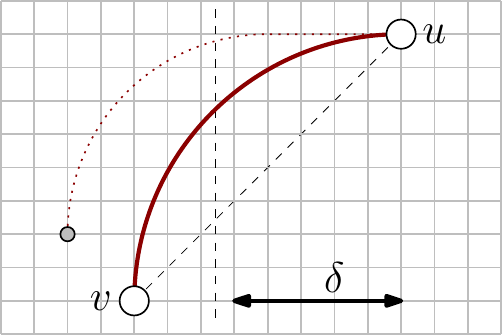}
	\hfil
	\includegraphics[width=.25\textwidth,page=7]{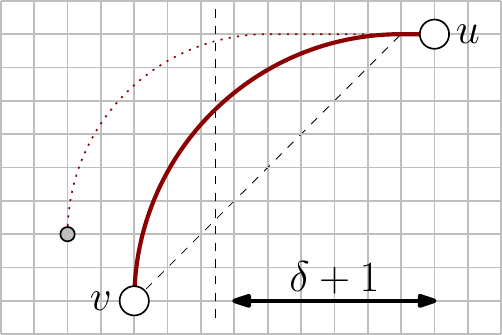}
	\caption{%
	Stretching an edge containing a vertical segment of length $\delta$.}
	\label{fig:h-stretch}
\end{figure}

\rephrase{Theorem}{\ref{theo:evenNicerSC2}}{\evenNicerSCTwo}
\begin{proof}
The time complexity follows from the fact that in Steps~\ref{s:h-stretch} and~\ref{s:v-stretch} of our algorithm we may have to stretch the drawing a linear number of times.

\begin{figure}[h!]
	\centering
	\includegraphics[scale=0.4,page=8]{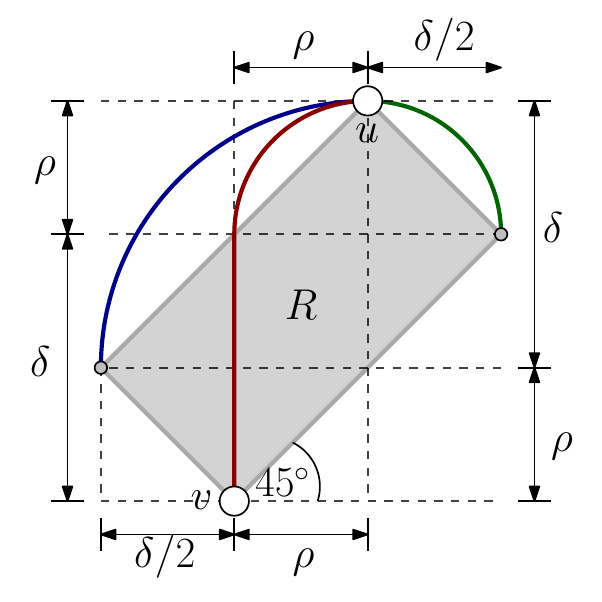}
	\caption{Proof of the area requirement.}
	\label{fig:areaRequirement}
\end{figure}

For the area requirement, we first consider Step~\ref{s:h-stretch}. Recall that in this step we stretch each red edge that is not a vertical line segment. Note that the red edges form a tree, which implies that we stretch at most $n-1$ edges. If at some point we stretch due to a red edge $(u,v)$ such that vertex $u$ is to the right of $v$ and above $v$, our algorithm guarantees that we will not stretch again due to another edge $(u,w)$ with $u$ to the right of $w$ and above $w$. This implies that the blue and the green edges incident to $u$ define a rectangle $R$ that will not contain another edge that must be stretched later in the algorithm (refer to the gray colored rectangle in Fig.~\ref{fig:areaRequirement}). If $(u,v)$ is composed of a vertical segment of length $\delta$ and a quarter circle arc of radius~$\rho$, then the entire drawing is stretched by $\delta$. This fixes the length of all segments necessary to compute the area of rectangle $R$; see Fig.~\ref{fig:areaRequirement}. In total, the area of $R$ is $\mathcal{A}(R)=\frac{1}{2}\delta^2 + \rho\delta$. Obviously, $\mathcal{A}(R)$ is minimized when $\rho=1$, for a constant $\delta$, which reduces the formula to $\mathcal{A}(R)=\frac{1}{2}\delta^2 + \delta$, which grows quadratically in $\delta$. Since the area of the drawing computed before the application of Step~\ref{s:h-stretch} is $O(n^2)$, the rectangles of red edges are on average of $O(n)$ area. Note that by planarity no two such rectangles overlap. Since the area of each rectangle grows quadratically in the length of the vertical segment of its associated red edge, the drawing requires the most total stretching, if each red edge is stretched by the same amount. In this case, $1/2\delta^2 + \delta = O(n)$, which implies that $\delta = O(\sqrt{n})$. Hence, the drawing obtained by Step~\ref{s:h-stretch} of our algorithm has width $O(n^{3/2})$.

For each vertex in Step~\ref{s:v-stretch} of our algorithm, we may increase the height of the drawing obtained by Step~\ref{s:h-stretch} by an amount that is bounded by the width of the initial drawing, i.e., $O(n^{3/2})$. This leads to a total height of $O(n^{5/2})$ and therefore a total area bound of $O(n^4)$.\qed
\end{proof}

\section{Example Run of our Drawing Algorithm}
\label{app:example}

In this section, we describe an example run of our drawing algorithm from Section~\ref{sec:kandinsky} on the octahedron graph. Figs.~\ref{fig:exampleRun-1} and~\ref{fig:exampleRun-2} illustrate Steps~\ref{s:h-stretch} and~\ref{s:v-stretch}, respectively. In particular, Fig.~\ref{fig:exampleRun1} shows the output of our modification of the shift-method and the first vertical cut. Fig.~\ref{fig:exampleRun2} shows the result of the first horizontal stretching and the second vertical cut. Fig.~\ref{fig:exampleRun3} shows the result after the second horizontal stretching, which is also the output of Step~\ref{s:h-stretch}. 

\begin{figure}[h!]
	\centering
	\begin{minipage}[b]{0.3\textwidth}
		\centering
		\subfloat[\label{fig:exampleRun1} { }]{
		\includegraphics[width=\textwidth,page=11]{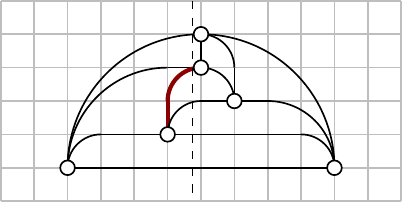}}
	\end{minipage}
	\hfill
	\begin{minipage}[b]{0.3\textwidth}
		\centering
		\subfloat[\label{fig:exampleRun2} { }]{
		\includegraphics[width=\textwidth,page=12]{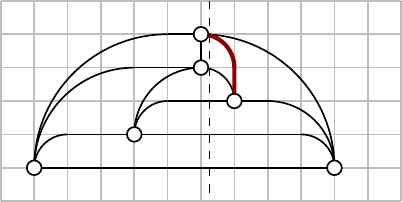}}
	\end{minipage}
	\hfill
	\begin{minipage}[b]{0.3\textwidth}
		\centering
		\subfloat[\label{fig:exampleRun3} { }]{
		\includegraphics[width=\textwidth,page=13]{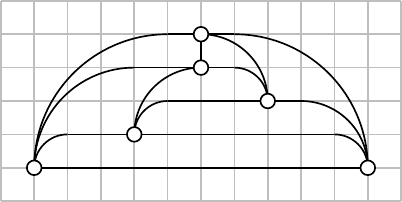}}
	\end{minipage}
	\caption{%
	Example run of Step~\ref{s:h-stretch} of our drawing algorithm.}
	\label{fig:exampleRun-1}
\end{figure}

Fig.~\ref{fig:exampleRun-2} illustrates how we assign new $y$-coordinates to the vertices in Step~\ref{s:v-stretch} of our algorithm. In particular, Fig.~\ref{fig:exampleRun4} shows how this is done for the first three vertices. Figs.~\ref{fig:exampleRun5},~\ref{fig:exampleRun6} and~\ref{fig:exampleRun7} illustrate how the fourth, the fifth and the sixth vertex of the octahedron is attached to the drawing. The bold edges in each subfigure of Fig.~\ref{fig:exampleRun-2} are the ones defining the maximum (denoted by $(v_k,w^\ast)$)  in the description of the algorithm. The final drawing is the one of Fig.~\ref{fig:exampleRun7}.

\begin{figure}[t!]
	\centering	
	\begin{minipage}[b]{0.3\textwidth}
		\centering
		\subfloat[\label{fig:exampleRun4} { }]{
		\includegraphics[width=\textwidth,page=14]{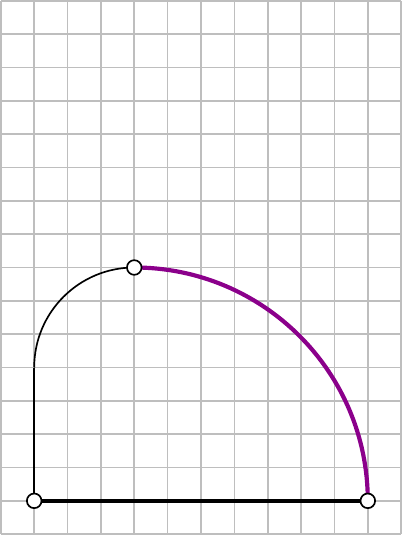}}
	\end{minipage}
	\hfil
	\begin{minipage}[b]{0.3\textwidth}
		\centering
		\subfloat[\label{fig:exampleRun5} { }]{
		\includegraphics[width=\textwidth,page=15]{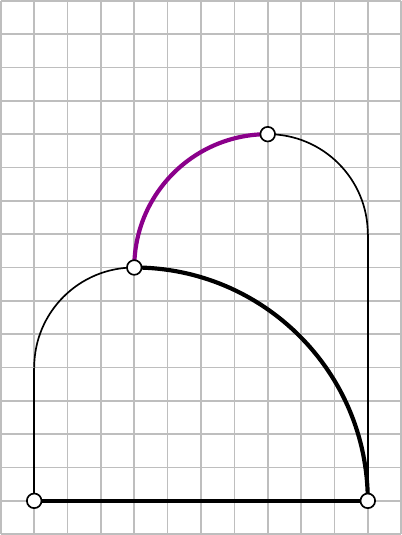}}
	\end{minipage}
	
	\begin{minipage}[b]{0.3\textwidth}
		\centering
		\subfloat[\label{fig:exampleRun6} { }]{
		\includegraphics[width=\textwidth,page=16]{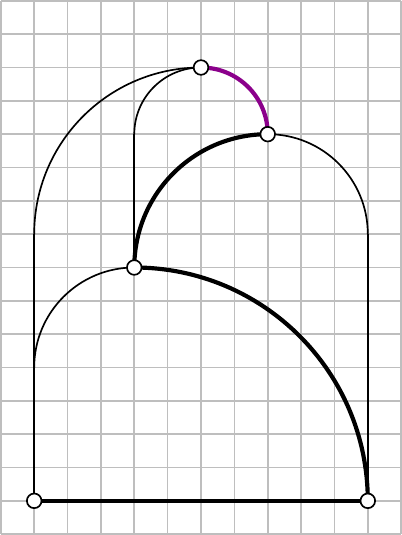}}
	\end{minipage}
	\hfil	
	\begin{minipage}[b]{0.3\textwidth}
		\centering
		\subfloat[\label{fig:exampleRun7} { }]{
		\includegraphics[width=\textwidth,page=17]{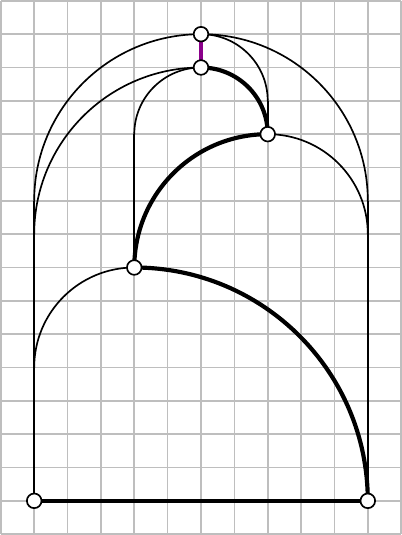}}
	\end{minipage}
	\caption{%
	Example run of Step~\ref{s:v-stretch} of our drawing algorithm.}
	\label{fig:exampleRun-2}
\end{figure}}%
\end{document}